\documentclass{article}
\usepackage{a4wide}
\usepackage{amsmath,amsthm,amssymb}
\usepackage{graphicx}
\usepackage{tikz}
\usepackage{hyperref}
\usepackage{multirow}

\DeclareMathOperator{\F}{\mathbb F}

\theoremstyle{plain}
\newtheorem{lemma}{Lemma}

\newcommand{\mF}{\mathcal F}
\newcommand{\set}[1]{\left\{{#1}\right\}}

\newcommand{\E}[2]{\mathbf E_{{#1}}\left[{#2}\right]}

\DeclareMathOperator{\sgn}{sgn}

\begin{document}

\title{Fast Sequential Decoding of Polar Codes}
\author{Peter Trifonov,  Vera Miloslavskaya, Ruslan Morozov \thanks{The authors are with the Distributed Computing and Networking Department, Saint Petersburg Polytechnic University. E-mail: petert@dcn.icc.spbstu.ru}}

\date{\today}
\maketitle

\begin{abstract}
 A new score function is proposed for stack decoding of polar codes, which enables one to accurately compare paths of different lengths. The proposed score function includes bias, which reflects the average behaviour of the correct path. This enables significant complexity reduction with respect to the original stack algorithm at the expense of  a negligible performance loss.  
\end{abstract}
\section{Introduction}
Polar codes were recently shown to be able to achieve the symmetric capacity of memoryless channels, while having low-complexity construction, encoding and decoding algorithms \cite{arikan2009channel}. However, the performance of polar codes of moderate length is substantially worse compared to LDPC and turbo codes used in today communication systems. This is both due to suboptimality of the classical successive cancellation (SC) decoding method, and poor minimum distance of polar codes. List decoding algorithm introduced in \cite{tal2015list} enables one to implement near maximum likelihood decoding of polar codes with complexity $O(Ln\log n)$, where $n$ is code length and $L$ is list size. Furthermore, the performance of polar codes concatenated with a CRC outer code \cite{tal2015list} and polar subcodes \cite{trifonov2016polar,trifonov2017randomized} under list decoding appears to be  better compared to known LDPC and turbo codes.

However, the complexity of the Tal-Vardy list decoding algorithm turns out to be rather high. It can be reduced by employing sequential decoding techniques \cite{niu2012stack,trifonov2013polar,miloslavskaya2014sequential}. These methods avoid construction of many useless low-probability paths in the code tree. Processing of such paths constitutes the most of the computational burden  of the list decoder. In this paper we show that by careful weighting of paths of different length, one can significantly reduce the computational complexity of the decoder.
 The proposed  path score function aims to estimate the conditional probability of the most likely codeword of a polar code, which may be obtained as a continuation of the considered path in the code tree.
It turns out that such function can be well approximated by the path score of the min-sum list SC decoder biased by its expected value. 
 Simulation results indicate that the proposed approach results in significant reduction of the average number of iterations performed by the decoder.

The paper is organized as follows. Section \ref{sBG} provides some background on polar codes and polar subcodes. Section \ref{sSeqDec} introduces the proposed sequential decoding method. Its improvements are discussed in Section \ref{sImprovements}. Simulation results illustrating the performance and complexity of the proposed algorithm are provided in Section \ref{sNumRes}. Finally, some conclusions are drawn.

\section{Polar codes}
\label{sBG}
\subsection{Code construction}
$(n=2^m,k)$, or $(n,k,d)$ polar code over $\F_2$ is a linear block code generated by $k$ rows of 
matrix $A_m=B_mF^{\otimes m}$, where $B_m$ is the bit reversal permutation matrix,  $F=\begin{pmatrix}1&0\\1&1
\end{pmatrix}$, and $\otimes m$ denotes 
$m$-times Kronecker product of the matrix with itself, and $d$ is the minimum distance of the code \cite{arikan2009channel}.    Hence, a codeword of a classical polar code is obtained as 
$c_0^{n-1}=u_0^{n-1}A_{m}$, where $a_s^t=(a_s,\dots,a_t)$, $u_i=0$ for $i\in \mathcal F,$  $\mathcal F\subset\set{0,\ldots,n-1}$ is the 
set of $n-k$ frozen symbol indices, and the remaining symbols of the information vector $u_0^{n-1}$ are set to the data symbols being encoded.

Let $U_0,\dots,U_{n-1}$ and $Y_0,\dots,Y_{n-1}$ be the random variables corresponding to the input symbols of the polarizing transformation $A_m$,  and output symbols of a memoryless output-symmetric channel, respectively.
It is possible to show that matrix $A_m$
transforms the original binary input memoryless output-symmetric channel
$W_0^{(0)}(Y|C)=W(Y|C)$ into bit subchannels  
$W_m^{(i)}(\mathbf Y,U_0^{i-1}|U_i),
$ 
the capacities of these subchannels converge with $m$ to $0$ or $1$, and the fraction
of subchannels with capacity close to $1$ converges to the capacity of $W_0^{(0)}(Y|C)$.
Here $\mathbf Y=Y_0^{n-1}$, and $C\in\F_2$ is the random variable corresponding to channel input.

The conventional approach to construction of an $(n,k)$ polar code assumes that $\mF$ is the set of  $n-k$ indices $i$ of bit subchannels $W_m^{(i)}(Y_0^{n-1},U_0^{i-1}|U_i)$ with the highest error probability.
 It was suggested in \cite{trifonov2016polar} to set some frozen symbols $U_i, i\in \mF,$ not to zero, but to linear combinations of  other symbols,
 i.e. the random variables should satisfy
\begin{equation} 
\label{mDynFrozen}
 U_i=\sum_{s=0}^{i-1}V_{j_i,s}U_s,
\end{equation} 
  where $V$ is a $(n-k)\times n$ binary matrix, such that its rows have last non-zero values in distinct columns, and $j_i$ is the index of the row having the last non-zero element in column $i$. 
Alternatively, these constraints can be stated as $U_0^{n-1}V^T=0$. This special shape of matrix $V$ enables one to implement decoding of polar subcodes using straightforward generalizations of the successive cancellation algorithm and its extensions. 

The obtained codes are referred to as polar subcodes. Polar codes with CRC \cite{tal2015list} can be considered as a special case of polar subcodes. Polar subcodes were shown to provide substantially better performance compared to classical polar codes under list decoding.  Therefore, we derive the decoding algorithm for the case of polar subcodes.

\subsection{The successive cancellation decoding algorithm}
Decoding of polar (sub)codes can be implemented by the successive cancellation (SC) algorithm.
It is  convenient to describe the SC algorithm in terms of probabilities $W_{m}^{(i)}\set{U_0^i=v_0^{i}|\mathbf Y=y_0^{n-1}}=W_{m}^{(i)}\set{v_0^{i}|y_0^{n-1}}$   of transmission of various vectors $v_0^{n-1}A_m$ with given values $v_0^i$, provided that the receiver observes a noisy vector $y_0^{n-1}$, i.e.
\begin{align}
\lefteqn{W_{m}^{(i)}\set{v_0^{i}|y_0^{n-1}}=\frac{W_m^{(i)}(y_0^{n-1},v_0^{i-1}|v_i)}{2W(y_0^{n-1})}}\nonumber\\
=&\sum_{v_{i+1}^{n-1}}W_{m}^{(n-1)}\set{v_0^{n-1}|y_0^{n-1}}
=\sum_{v_{i+1}^{n-1}}\prod_{j=0}^{n-1}W\set{(v_0^{n-1}A_{m})_j|y_j}
\label{mTotalProb}
\end{align}

 At phase $i$ the SC decoder  makes decision 
\begin{equation}
\label{mSCDecisionRule}
\widehat{u}_i=\begin{cases} \arg \max_{v_i\in\F_2} W_{m}^{(i)}\set{\widehat u_0^{i-1}.v_i|y_0^{n-1}},& i\not \in \mathcal F \\\sum_{s=0}^{i-1}V_{j_i,s}\widehat u_s,&\text{otherwise},\end{cases}
\end{equation}
where $a.b$ denotes a vector obtained by appending $b$ to $a$.  
The probabilities 
used in this expression can be recursively computed as 
\begin{align}
\label{mSCProb1}
\begin{split}
\lefteqn{W_{\lambda}^{(2i)}\set{{v}_0^{2i}|y_0^{2^{\lambda}-1}}=}&\\
&\displaystyle\sum_{v_{2i+1}} W_{\lambda-1}^{(i)}
\set{{v}_{0,e}^{2i+1} \oplus { v}_{0,o}^{2i+1}|y_0^{2^{\lambda-1}-1}}W_{\lambda-1}^{(i)}\set{ v_{0,o}^{2i+1} |y_{2^{\lambda-1}}^{2^{\lambda}-1}}
\end{split}
\\
\label{mSCProb2}
\begin{split}
\lefteqn{W_{\lambda}^{(2i+1)}\set{{v}_0^{2i+1}|y_0^{2^{\lambda}-1}}=}&\\
& W_{\lambda-1}^{(i)}
\set{{v}_{0,e}^{2i+1} \oplus { v}_{0,o}^{2i+1}|y_0^{2^{\lambda-1}-1}}W_{\lambda-1}^{(i)}\set{ v_{0,o}^{2i+1} |y_{2^{\lambda-1}}^{2^{\lambda}-1}},
\end{split}
\end{align}
where $ 0< \lambda\leq m,$  $a_{0,e}^t$ and $a_{0,o}^t$ denote subvectors of $a_0^t$, consisting of elements with even and odd indices, respectively.
It is convenient to implement these calculations in LLR domain as
\begin{align}
\label{mLLRDecoder1}
{L_{\lambda}^{(2i)}(v_0^{2i-1}|y_0^{2^{\lambda}-1})=}&\nonumber\\=2\tanh^{-1}&\left(\tanh\left(\frac{L_{\lambda-1}^{(i)}(v_{0,e}^{2i-1} \oplus { v}_{0,o}^{2i-1}|y_0^{2^{\lambda-1}-1})}{2}\right)\right.\nonumber\\
&\times\left.\tanh \left(\frac{L_{\lambda-1}^{(i)}(v_{0,o}^{2i-1}|y_{2^{\lambda-1}}^{2^{\lambda}-1})}{2}\right)\right), \\
\label{mLLRDecoder2}
L_{\lambda}^{(2i+1)}(v_0^{2i}|y_0^{2^{\lambda}-1})=&(-1)^{u_{2i}}L_{\lambda-1}^{(i)}\left({v}_{0,e}^{2i-1} \oplus { v}_{0,o}^{2i-1}|y_0^{2^{\lambda-1}-1}\right)\nonumber\\&+L_{\lambda-1}^{(i)}\left({ v}_{0,o}^{2i-1}|y_{2^{\lambda-1}}^{2^{\lambda}-1}\right),
\end{align}
where $L_{\lambda}^{(i)}(v_0^{i-1}|y_0^{2^{\lambda}-1})=\log\frac{W_\lambda^{(i)}\set{v_0^{i-1}.0|y_0^{2^{\lambda}-1}}}{W_\lambda^{(i)}\set{v_0^{i-1}.1|y_0^{2^{\lambda}-1}}}$, $0\leq i<n,0\leq \lambda<m$, so that the decision rule for $i\notin \mF$ becomes $$\widehat u_i=\begin{cases} 0,&L_{m}^{(i)}(\widehat u_0^{i-1}|y_0^{n-1})>0\\
1,&\text{otherwise.}\end{cases}$$

\subsection{Improved decoding algorithms}
The original SC decoding algorithm does not provide maximum likelihood decoding. A successive cancellation list (SCL) decoding algorithm was suggested in \cite{tal2015list},
and shown to achieve substantially better performance with complexity $O(Ln\log n)$. However, large values of $L$ are needed in order to implement near-ML\ decoding of polar subcodes and polar codes with CRC. This makes practical implementations of such list decoders very challenging. 

On the other hand, one does not need in practice to obtain a list of codewords, but just a single most probable one. The original Tal-Vardy algorithm for polar codes with CRC examines the elements in the obtained list, and discards those with invalid checksums. This algorithm can be easily tailored to process the dynamic freezing constraints \eqref{mDynFrozen} before the decoder reaches the last phase, so that the output list contains only valid codewords of a polar subcode. However, even in this case one has to discard  $L-1$ codewords from the obtained list, so most of the computational work performed by the Tal-Vardy decoder is just wasted. 

This problem was addressed in \cite{niu2012stack}, where a generalization of the stack algorithm to the case of polar codes was suggested. It provides lower average decoding complexity compared to the Tal-Vardy algorithm. In this paper we revise the stack decoding algorithm for polar (sub)codes, and show that its complexity can be substantially reduced.

\section{Path score}
\label{sSeqDec}
\subsection{Stack decoding algorithm}
\label{sSeqDecAlg}
Let $u_0^{n-1}$ be the information vector used by the transmitter. Given a received noisy vector $y_0^{n-1}$, the proposed decoding algorithm constructs sequentially a number of partial candidate information vectors $v_0^{\phi-1}\in \F_2^\phi, \phi\leq n$, evaluates how close their continuations $v_0^{n-1}$ may be to the received sequence, and eventually produces a single codeword, being a solution of the decoding problem.

It is convenient to represent the set of information vectors as a tree. The nodes of the tree correspond to  vectors $v_0^{\phi-1}, 0\leq \phi<n$, satisfying \eqref{mDynFrozen}. At depth $\phi$, each node $v_0^{\phi-1}$ has  two children $v_0^{\phi-1}.0$ and $v_0^{\phi-1}.1$.  The root of the tree corresponds to an empty vector.
 By abuse of notation, the path from the root of the tree to a node $v_0^{\phi-1}$ is also denoted by $v_0^{\phi-1}$.
A decoding algorithm for a polar (sub)code  needs to consider only valid paths, i.e. the paths satisfying constraints \eqref{mDynFrozen}.

The stack  decoding algorithm  \cite{Zigangirov1966some,johannesson1998fundamentals,niu2012stack,miloslavskaya2014sequential}
employs a  priority queue\footnote{A PQ is commonly called "stack" in the sequential decoding literature. However, the implementation of the considered algorithm relies on Tal-Vardy data structures \cite{tal2015list}, which make use of the true stacks. Therefore, we employ the standard terminology of computer science.} (PQ) to store paths together with their scores.
A PQ is a data structure, which contains tuples $(M,v_0^{\phi-1})$, where $M=M(v_0^{\phi-1},y_0^{n-1})$ is the score of path $v_0^{\phi-1}$,  and provides efficient algorithms for the following operations \cite{Cormen2001introduction}:
\begin{itemize}
\item push a tuple into the PQ;
\item pop  a tuple $(M,v_0^{\phi-1})$ (or just $v_0^{\phi-1}$) with the highest $M$;
\item remove  a given tuple from the PQ.
\end{itemize}
We assume here that the PQ may contain at most $D$ elements.

In the context of polar codes, the stack decoding algorithm operates as follows:
\begin{enumerate}
\item Push into the PQ the root of the tree with score $0$. Let $t_0^{n-1}=0$.
\item Extract from the PQ a path $v_0^{\phi-1}$ with the highest score. Let $t_\phi\gets t_\phi+1$.
\item If $\phi=n$, return codeword $v_0^{n-1}A_{m}$ and terminate the algorithm.
\item If  the number of valid children of path $v_0^{\phi-1}$ exceeds the amount of free space in the PQ, remove from it the element with the smallest score.
\item Compute the scores $M(v_0^\phi,y_0^{n-1})$ of  valid children $v_0^\phi$ of the extracted path, and push them into the PQ. 
\item If $t_{\phi}\geq L$, remove from PQ all paths $v_0^{j-1}, j\leq \phi$.
\item Go to step 2.
\end{enumerate}
In what follows, one iteration means one pass of the above algorithm over steps 2--7. Variables $t_\phi$ are used to ensure that the worst-case complexity of the algorithm does not exceed that of a list SC decoder with list size $L$.

The parameter $L$ has the same impact on the performance of the decoding algorithm as the list size in the Tal-Vardy list decoding algorithm, since it imposes an upper bound on number of paths $t_\phi$ considered by the decoder at each phase $\phi$. Step 6 ensures that the algorithm terminates in at most $Ln$ iterations. This is also an upper bound on the number of entries stored in the PQ. However, the algorithm can work with PQ of much smaller size $D$. Step 4 ensures that this size is never exceeded.
\subsection{Score function}
There are many possible ways to define a score function for sequential decoding. In general, this should be done so that one can perform meaningful comparison of paths $v_0^{\phi-1}$ of different length $\phi$. The classical Fano metric for sequential decoding of convolutional codes is given by $$P\set{\mathcal M|y_0^{n-1}}=\frac{P\set{\mathcal M,y_0^{n-1}}}{\prod_{i=0}^{n-1}W(y_i)},$$ where $\mathcal M$ is a variable-length message (i.e. a path in the code tree), and $W(y_i)$ is the probability measure induced on the channel output
alphabet when the channel inputs follow some prescribed (e.g. uniform) distribution \cite{massey1972variable}. In the context of polar codes, a straightforward implementation of this approach would correspond to score function $$M_1(v_0^{\phi-1},y_0^{n-1})=\log W_m^{(\phi-1)}\set{v_0^{\phi-1}|y_0^{n-1}}.$$
This is exactly the score function used in \cite{niu2012stack}.  However, there are several shortcomings in such definition:
\begin{enumerate}
\item Although the value of the score does depend on all $y_i, 0\leq i<n$, it does not take into account freezing constraints on symbols $u_i, i\in \mF,i\geq \phi$. As a result, there may exist  incorrect paths $v_0^{\phi-1}\neq u_0^{\phi-1}$, which have many low-probability continuations $v_0^{n-1},v_\phi^{n-1}\in \F_2^{n-\phi}$, such that the probability $$W_m^{(\phi-1)}\set{v_0^{\phi-1}|y_0^{n-1}}=\sum_{v_\phi^{n-1}}W_m^{(n-1)}\set{v_0^{n-1}|y_0^{n-1}}$$
becomes quite high, and the stack decoder is forced to expand such a path. Note that this is not a problem in the case of convolutional codes, where the decoder may recover after an error burst, i.e. obtain a codeword, which is identical to the transmitted one, except for a few closely located symbols. \item Due to freezing constraints, not all vectors $v_0^{\phi-1}\in\F_2^\phi$ correspond to valid paths in the code tree. This does not allow one to fairly compare the probabilities of paths of different lengths, which include different number of frozen symbols.
\item  Computing probabilities $W_m^{(\phi-1)}\set{v_0^{\phi-1}|y_0^{n-1}}$ involves expensive multiplications and is prone to numeric errors. 
\end{enumerate}

The first of the above problems can be addressed by considering only the most probable continuation of path $v_0^{\phi-1}$, i.e. the score function can be defined as
\begin{equation}
\label{mSubtreeMax}
M_2(v_0^{\phi-1},y_0^{n-1})=\max_{v_\phi^{n-1}\in \F_2^{n-\phi}}\log W_m^{(n-1)}\set{v_0^{n-1}|y_0^{n-1}}.
\end{equation}
Observe that maximization is performed over last $n-\phi$ elements of vector $v_0^{n-1}$, while the remaining ones are given by $v_0^{\phi-1}$.
 Let us further define $$\mathbf V(v_0^{\phi-1},y_0^{n-1})=\arg\max_{\substack{w_0^{n-1}\in \F_2^{n}\\w_0^{\phi-1}=v_0^{\phi-1}}}\log W_m^{(n-1)}\set{w_0^{n-1}|y_0^{n-1}},$$
i.e. $M_2(v_0^{\phi-1},y_0^{n-1})=\log W_m^{(n-1)}\set{\mathbf V(v_0^{\phi-1},y_0^{n-1})|y_0^{n-1}}.$
 
As it will be shown below, employing such score function already provides significant reduction of the average number of iterations at the expense of a negligible performance degradation. Furthermore, it turns out that this score  is exactly equal to the one used in the min-sum version of the Tal-Vardy list decoding algorithm \cite{balatsoukasstimming2015llrbased}, i.e. it can be computed in a very simple way. 

To address the second problem, we need to evaluate the probabilities of vectors $v_0^{n-1}$ under freezing conditions. To do this,  consider the set of valid length-$\phi$ prefixes of the input vectors of the polarizing transformation, i.e.   $$C(\phi)=\set{v_0^{\phi-1}\in \F_2^n|v_i=\sum_{s=0}^{i-1}V_{j_i,s}v_s,i\in \mF,0\leq i<\phi}.$$ 
Let us further define the set of their most likely continuations, i.e. $$\overline  C(\phi)=\set{\mathbf V(v_0^{\phi-1},y_0^{n-1})|v_0^{\phi-1}\in C(\phi)}.$$

For any $v_0^{n-1}\in \overline C(\phi)$ the probability of transmission of  $v_0^{n-1}A_m$, under condition  of $v_0^{\phi-1}\in  C(\phi)$ and   given the received vector $y_0^{n-1}$, equals 
\begin{align*}
{\mathbb W\set{v_0^{n-1}|y_0^{n-1},C(\phi)}=}
\frac{W_m^{(n-1)}\set{U_0^{n-1}=v_0^{n-1}| y_0^{n-1}}}{W_m^{(n-1)}\set{U_0^{n-1}\in \overline C(\phi)| y_0^{n-1}}}.
\end{align*}

An ideal score function could be defined as $$\mathbb M(v_0^{\phi-1},y_0^{n-1})=\log \mathbb W\set{\mathbf V(v_0^{\phi-1},y_0^{n-1})|y_0^{n-1},C(\phi)}.$$
Observe that this function is defined only for vectors $v_0^{\phi-1}\in C(\phi)$, i.e. those satisfying freezing constraints up to phase $\phi$.

Unfortunately, there is no simple and obvious way to compute $\pi(\phi,y_0^{n-1})=W_m^{(n-1)}\set{U_0^{n-1}\in \overline C(\phi)|y_0^{n-1}}$. Therefore, we have to develop an approximation. 

  It can be seen that 
\begin{align}
\label{mCodeProb}
\pi(\phi,y_0^{n-1})=&   W_m^{(n-1)}\set{\mathbf V(u_0^{\phi-1})| y_0^{n-1}}+\nonumber\\&\underbrace{\sum_{\substack{v_0^{\phi-1}\in C(\phi) \\v_0^{\phi-1}\neq{u_0^{\phi-1}}}}W_m^{(n-1)}\set{\mathbf V(v_0^{\phi-1})|y_0^{n-1}}}_{\mu(u_0^{\phi-1},y_0^{n-1})}.
\end{align}
Observe that $p=\E{\mathbf Y}{\frac{\mu(u_0^{n-1},\mathbf Y)}{\pi(\phi,\mathbf Y)}}$ is the probability of the min-sum version of the Tal-Vardy list decoding algorithm with infinite list size not obtaining $u_0^{\phi-1}$ as the most probable path at phase $\phi$.  We consider decoding of polar (sub)codes, which are constructed to have low list SC decoding error probability even for small list size in the considered channel $W(y|c)$. Hence, it can be assumed that $p\ll 1$. This implies that with high probability $\mu(u_0^{\phi-1},y_0^{n-1})\ll W_m^{(\phi-1)}\set{\mathbf V(u_0^{\phi-1})| y_0^{n-1}}$, i.e. $\pi(\phi,y_0^{n-1})\approx    W_m^{(\phi-1)}\set{\mathbf V(u_0^{\phi-1})| y_0^{n-1}}$. 

   However, a real decoder cannot compute this value, since the transmitted vector $u_0^{n-1}$ is not available at the receiver side.  Hence, we  propose to further approximate the logarithm of the first term in \eqref{mCodeProb} with its expected value over $\mathbf Y$, i.e.
$$\log \pi(\phi,y_0^{n-1})\approx \Psi(\phi)=\E{\mathbf Y}{\log W_m^{(\phi-1)}\set{\mathbf V(u_0^{\phi-1})|\mathbf Y}}$$
Observe that this value depends only on $\phi$ and underlying channel $W(y|c)$, and can be pre-computed offline.

Hence, instead of the ideal score function $\mathbb M(v_0^{\phi-1},y_0^{n-1})$ we propose to use an approximate one 
\begin{equation}
\label{mPathScore}
M_3(v_0^{\phi-1},y_0^{n-1})=M_2(v_0^{\phi-1},y_0^{n-1})-\Psi(\phi).
\end{equation}

\subsection{Computing the score function}
Consider  computing  $$R_m^{(\phi-1)}(v_0^{\phi-1},y_0^{n-1})=M_2(v_0^{\phi-1},y_0^{n-1}).$$  
Let  the modified log-likelihood ratios be defined as
\begin{equation}
\label{mLLRS}
S_m^{(\phi)}(v_0^{\phi-1}|y_0^{n-1})=R_m^{(\phi)}(v_0^{\phi-1}.0,y_0^{n-1})-R_m^{(\phi)}(v_0^{\phi-1}.1,y_0^{n-1}).
\end{equation}
It can be seen that
\begin{align}
R_m^{(\phi)}(v_0^{\phi},y_0^{n-1})=&
R_m^{(\phi-1)}(v_0^{\phi-1},y_0^{n-1})\nonumber\\&+\tau(S_m^{(\phi)}(v_0^{\phi-1}|y_0^{n-1}),v_\phi), \label{mLogProb}
\end{align}
where 
$$\tau(S,v)=\begin{cases}
0,&\text{if $\sgn(S)=(-1)^v$}\\
-|S|,&\text{otherwise.}
\end{cases}$$
is the penalty function.

Indeed, let $\tilde v_0^{n-1}=\mathbf V(v_0^{\phi-1})$. If $v_\phi=\tilde v_\phi$, then the most probable continuations of  $v_0^{\phi-1}$ and $v_0^{\phi}$ are identical.
Otherwise, $-\left|S_m^{(\phi)}(v_0^{\phi-1}|y_0^{n-1})\right|$ is exactly the difference between the log-probability of the most likely continuations of $v_0^{\phi-1}$ and $v_0^{\phi}$. 

The initial value for recursion \eqref{mLogProb} is given by $$R_m^{(-1)}(y_0^{n-1})=\log \prod_{i=0}^{n-1}W\set{C=\hat c_i|Y=y_i},$$
where $\hat c_i$ is the hard decision  corresponding to $y_i$. However, this value can be replaced with $0$, since it does not affect the selection of paths in the stack algorithm.

In order to obtain a simple expression for the proposed  score function, observe that   
\begin{align*}
W_m^{(n-1)}\set{ v_0^{n-1}|y_0^{n-1}}=&W_{m-1}^{(n/2-1)}\set{v_{0,e}^{n-1}\oplus v_{0,o}^{n-1}|y_0^{\frac{n}{2}-1}}\cdot\\& W_{m-1}^{(n/2-1)}\set{ v_{0,o}^{n-1}|y_{\frac{n}{2}}^{n-1}}.
\end{align*}
Since $R_m^{(\phi-1)}(v_0^{\phi-1},y_0^{n-1})$ is obtained by maximization of $W_m^{(n-1)}\set{ v_0^{n-1}|y_0^{n-1}}$ over $v_\phi^{n-1}$, it can be seen that 
\begin{align*} 
\lefteqn{R_{\lambda}^{(2i)}(v_0^{2i},y_0^{N-1})=}&\nonumber\\
&\max_{v_{2i+1}} \left(R_{\lambda-1}^{(i)}\left(v_{0,e}^{2i+1}\oplus v_{0,o}^{2i+1},y_0^{\frac{N}{2}-1}\right)+R_{\lambda-1}^{(i)}\left(v_{0,o}^{2i+1},y_{\frac{N}{2}}^{N-1}\right)\right),\\
\lefteqn{R_{\lambda}^{(2i+1)}(v_0^{2i+1}|y_0^{N-1})=}&\nonumber\\
&\quad \quad \quad R_{\lambda-1}^{(i)}\left(v_{0,e}^{2i+1}\oplus v_{0,o}^{2i+1},y_0^{\frac{N}{2}-1}\right)+R_{\lambda-1}^{(i)}\left(v_{0,o}^{2i+1},y_{\frac{N}{2}}^{N-1}\right),
\end{align*}
where $N=2^\lambda,$ $0< \lambda\leq m$, and initial values for these recursive expressions are given by $R_0^{(0)}(b,y_j)=\log W_0^{(0)}\set{b|y_j}$, $b\in\set{0,1}$. 
From \eqref{mLLRS} one obtains 
\begin{align*}
S_{\lambda}^{(2i)}(v_0^{2i-1}|y_0^{2^{\lambda}-1})=&\max(J(0)+K(0),J(1)+K(1))-\nonumber\\&\max(J(1)+K(0),J(0)+K(1))\nonumber\\=
&\max(J(0)-J(1)+K(0)-K(1),0)-\nonumber\\&\max(K(0)-K(1),J(0)-J(1))\\
S_{\lambda}^{(2i+1)}(v_0^{2i}|y_0^{2^{\lambda}-1})=&J(v_{2i})+K(0)-J(v_{2i}+1)-K(1)
\end{align*}
where $J(c)=R_{\lambda-1}^{(i)}((v_{0,e}^{2i-1}\oplus v_{0,o}^{2i-1}).c|y_0^{2^{\lambda-1}-1})$, $K(c)=R_{\lambda-1}^{(i)}(v_{0,o}^{2i-1}.c|y_{2^{\lambda-1}}^{2^{\lambda}-1})$. Observe that $$J(0)-J(1)=a=S_{\lambda-1}^{(i)}(v_{0,e}^{2i-1}\oplus v_{0,o}^{2i-1}|y_0^{2^{\lambda-1}-1})$$ and $$K(0)-K(1)=b=S_{\lambda-1}^{(i)}(v_{0,o}^{2i-1}|y_{2^{\lambda-1}}^{2^\lambda-1})$$

It can be obtained from these expressions that the modified log-likelihood ratios are given by 
\begin{align}
\label{mMinSum1}
S_{\lambda}^{(2i)}(v_0^{2i-1}|y_0^{2^\lambda-1})=&Q(a,b)=\sgn (a)\sgn (b)\min(|a|,|b|),\\
\label{mMinSum2}
S_{\lambda}^{(2i+1)}(v_0^{2i}|y_0^{2^\lambda-1})=&P(v_{2i},a,b)=(-1)^{v_{2i}}a+b.
\end{align}

The initial values for this recursion are given by $S_0^{(0)}(y_i)=\log\frac{W\set{0|y_i}}{W\set{1|y_i}}$.

These expressions can be readily recognized as the min-sum approximation for \eqref{mLLRDecoder1}--\eqref{mLLRDecoder2}, and \eqref{mLogProb} coincides with an approximation for $M_1(v_0^\phi,y_0^{n-1})$  \cite{leroux2013semiparallel,leroux2011hardware,balatsoukasstimming2015llrbased}. However, these
are also the exact values, which reflect the probability of the most likely continuation of a given path $v_0^{\phi-1}$ in the code tree. 

 Finally, we illustrate the  meaning of  $M_{2}(v_0^{n-1},y_0^{n-1})$. Let $S_0^{n-1}$ be the LLR vector corresponding to the received noisy sequence $y_0^{n-1}$.  Let $$E(c_0^{n-1},S_0^{n-1})=-\sum_{i=0}^{n-1}\tau(S_i,c_i)$$
 be the ellipsoidal weight (also known as correlation discrepancy) of vector $c_0^{n-1}\in \F_2^n$   \cite{Valembois2004box,moorthy1997softdecision}. It is possible to show that the ML decoding problem for the case of transmission of codewords of a code $\mathcal C$ over a memoryless channel can be formulated as $$\hat u_0^{n-1}=\arg \min_{c_0^{n-1}\in \mathcal C}E(c_0^{n-1},S_0^{n-1})$$
\begin{lemma}
\label{lWeightSum}
For any  $c_0^{2n-1}\in\F_2^{2n-1}$
$$E(c_0^{2n-1},S_0^{2n-1})=E(c_{0,e}^{2n-1}+c_{0,o}^{2n-1},\tilde S_0^{n-1})+E(c_{0,o}^{n-1},\overline S_0^{n-1}),$$
where $\tilde S_i= Q(S_{2i},S_{2i+1})$, $\overline S_i=P(c_{2i},S_{2i},S_{2i+1})$.
\end{lemma} 
\begin{proof}
It is sufficient to prove the statement for $n=1$. The result follows by examining all possible combinations of $c_0, c_1\in \F_2$.
\end{proof}
Applying Lemma \ref{lWeightSum} recursively and assuming $R_m^{(-1)}(y_0^{n-1})=0$, one obtains
that 
\begin{equation}
\label{mEWScore}
M_{2}(v_0^{n-1},y_0^{n-1})=-E(v_0^{n-1}A_m,\mathbf S),
\end{equation}
where $\mathbf S=(S_0^{(0)}(y_0),\dots,S_0^{(0)}(y_{n-1}))$. This means that the proposed score function reflects the ellipsoidal weight of the transmitted codeword.

\subsection{The bias function}
The function  $\Psi(\phi)$ is equal to the expected value of the logarithm of the probability of a length-$\phi$ part of the correct path, i.e. the path corresponding to the vector $u_0^{n-1}$ used by the encoder. Employing this function enables one to estimate how far a particular path $v_0^{\phi-1}$ has diverted from the expected behaviour of a correct path. The bias function can be computed offline under the assumption of zero codeword transmission. Indeed, in this case the cumulative  density functions $F_{\lambda}^{(i)}(x)$ of $S_\lambda^{(i)}$ are given by \cite{kern2014new}
\begin{align}
&F_{\lambda}^{(2i)}(x)=\begin{cases}
2F_{\lambda-1}^{(i)}(x)(1-F_{\lambda-1}^{(i)}(-x)),&x< 0\\
2F_{\lambda-1}^{(i)}(x)-(F_{\lambda-1}^{(i)}(-x))^2-(F_{\lambda-1}^{(i)}(x))^2,&x\geq 0
\end{cases}\label{mMinSumCumDens}\\
&F_{\lambda}^{(2i+1)}(x)=\int_{-\infty}^\infty F_{\lambda-1}^{(i)}(x-y) dF_{\lambda-1}^{(i)}(y),
\label{mDensityConv}
\end{align}
where $F_0^{(0)}(x)$ is the CDF of the channel output LLRs.
Then one can compute 
\begin{equation}
\label{mBias}
\Psi(\phi)=-\sum_{i=0}^\phi \int_{-\infty}^0F_{m}^{(i)}(x)dx.
\end{equation}

The bias function $\Psi(\phi)$ depends only on $m$ and channel properties, so it can be used for decoding of any polar (sub)code of a given length.

Figure \ref{fHeuristic} illustrates the bias function for the case of BPSK modulation and AWGN channel with different noise standard deviations $\sigma$. 
 
\begin{figure}
\centering\includegraphics[width=0.5\textwidth]{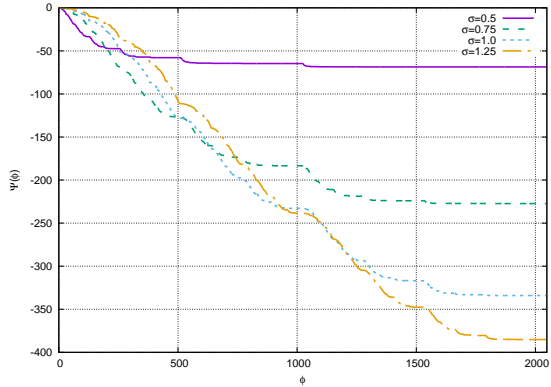}
\caption{The bias function for   AWGN  channel, $m=11.$}
\label{fHeuristic}
\end{figure}

\subsection{Summary of the proposed approach}
We propose to implement decoding of polar (sub)codes by employing the stack algorithm described in Section \ref{sSeqDecAlg}, making use of the score function $M_3(v_0^{\phi-1},y_0^{n-1})$ given by \eqref{mPathScore}. The first component of the score function is given by \eqref{mLogProb}, i.e. it is the accumulated penalty 
for phases $0,\dots,\phi-1$. Phase penalty is equal (up to the sign) to the absolute value of LLR $S_m^{(\phi)}(v_0^{\phi-1}|y_0^{n-1})$, if its sign does not agree with $v_i$, and zero otherwise. The LLRs are recursively computed using the min-sum method, as described in \eqref{mMinSum1}--\eqref{mMinSum2}. The second component of the score function, bias,  is the expected value of the penalty for these phases under the assumption that $v_0^{\phi-1}=u_0^{\phi-1}$, i.e. the considered path is correct. The bias function can be pre-computed offline according to \eqref{mBias} from the transition probability function $W(Y|C)$     of the underlying channel.

The parameters of the algorithm are effective list size $L$ and maximal size of the priority queue $D$, as well as channel transition probability function $W(Y|C)$, which is used to construct the bias function.

Observe that the proposed approach can be used not only with classical polar codes, but also polar subcodes \cite{trifonov2016polar} and polar codes with CRC \cite{tal2015list}. 
\subsection{Complexity analysis}
The algorithm presented in Section \ref{sSeqDecAlg} extracts from the PQ length-$\phi$ paths at most $L$ times. At each iteration it needs to calculate the LLR $S_m^{(\phi)}(v_0^{\phi-1}|y_0^{n-1})$. Intermediate values for these calculations can be reused in the same way as in \cite{tal2015list}. Hence, LLR calculations require at most $O(Ln\log n)$ operations. However, simulation results presented below suggest that the average complexity of the proposed algorithm is substantially lower, and at high SNR approaches $O(n\log n)$, the complexity of the SC algorithm.

\section{Improvements}
\label{sImprovements}
\subsection{List size adaptation}
\label{sLSA}
In the case of infinite size $\Theta$ of the priority queue the performance of the proposed algorithm depends mostly on parameter $L$. Setting $L=\infty$ results in near\footnote{Similarly to the case of sequential decoding of convolutional codes, employing a non-zero bias function may prevent the stack algorithm from providing ML decoding performance. } maximum likelihood decoding at the cost of excessive complexity and memory consumption.
  It was observed in experiments that in most cases
the decoding can be completed successfully with small $L$, 
and only rarely a high value of $L$  is required. 

It was also observed that  if the decoder needs to kill short paths too many times while processing a single noisy vector, this most likely means that the correct path is already killed, and decoding error is unavoidable. 

Therefore, in order to reduce the decoding complexity, we propose to change  list size $L$ adaptively. We propose to start decoding of a noisy vector with small $L$, and increase it, if the decoder is likely to have killed the correct path.  
In order to do this, we propose to keep track of the number of times $\kappa$ the algorithm actually removes some paths at step 6 of the algorithm presented in Section \ref{sSeqDecAlg}. If this value exceeds some threshold $\kappa_0$, then the decoding may need to be restarted with larger $L$, similarly to \cite{li2012adaptive}. 

However, more efficient approach is possible. In order to avoid repeating the same calculations, we propose not to kill paths permanently at step 6, but remove them from the PQ, and save the corresponding pairs $(\widetilde M,l)$ in a temporary array, where $l$ is an identifier of some path $v_0^{\phi-1}$,  $\widetilde M=M_{3}(v_0^{\phi-1},y_0^{n-1})$, and  $M$ is the score of the path extracted at step 2 of the algorithm presented in Section \ref{sSeqDecAlg}. Such paths are referred to as suspended paths. If $\kappa$ exceeds some threshold $\kappa_0$, instead of restarting the decoder, we propose to double list size $L$ and re-introduce into the PQ suspended paths with the score better than the score $M_0$ of the current path. This is performed until $L$ reaches some upper bound $L_{max}$. 
%
%
%
%
The value of $\kappa_0$ should be optimized by simulations.

\subsection{Early termination}
\label{sTermination}
The sequential decoder may kill the correct path at  step 6 of the algorithm presented  in Section \ref{sSeqDecAlg}.
In this case, the decoder cannot return the correct codeword, but  continues to perform useless calculations, inspecting many wrong paths in the code tree.
It is desired to detect quickly such case and abort decoding.

Observe that for a correct path $ u_0^{n-1}$ one has $\E{\mathbf Y}{M_{3}( u_0^{n-1},\mathbf Y)}=0$, and, if the  decoder does not make an error,  for any other path $v_0^{n-1}$ one has $M_3(v_0^{n-1},y_0^{n-1})<M_3(u_0^{n-1},y_0^{n-1})$.

At any phase $\phi$ the probability that  $v_0^{\phi-1}$ is a part of the correct path is an increasing function of  $M_{3}(v_0^{\phi-1},y_0^{n-1})$.
Therefore, we propose to abort decoding if the score $M_3(v_0^{\phi-1},y_0^{n-1})$ of a path $v_0^{\phi-1}$ extracted from the PQ at some iteration is below some threshold $T$. This may increase the decoding error probability. 
The threshold should be selected so that the probability of correct path $ u_0^{\phi-1}$ satisfying the described early termination criterion, i.e.  $M_{3}(u_0^{\phi-1},\mathbf Y)<T$, is sufficiently small. This requires one to study the distribution of $M_{3}( u_0^{\phi-1},\mathbf Y)$ at each phase $\phi$. Since this is a difficult problem, we propose to select $T$ based on the distribution of $\mu=M_{3}(u_0^{n-1},\mathbf Y)$. Namely, we propose to set $T$ to the value of  $p_{MAP}$-quantile of the distribution of  $\mu$, where $p_{MAP}$ is the codeword error probability of the MAP\ decoder.

It follows from \eqref{mEWScore} that $$M_{3}(u_0^{n-1},\mathbf Y)=\sum_{i=0}^{n-1}\left(\tau(S_i,c_i)-\E{Y_i}{\tau(S_i,c_i)}\right),$$
where $S_i=S_0^{(0)}(Y_i)$ and  $c_0^{n-1}= u_0^{n-1}A_m$. Assuming that  zero codeword was transmitted, i.e. $u_0^{n-1}=\mathbf 0$, one can derive the probability density function of  log-likelihood ratios $S_0^{(0)}(Y_i)$, and compute the PDF of $M_3(u_0^{n-1},\mathbf Y)$ as $n$-times convolution of the PDFs of $\tau(S_0^{(0)}(Y_i),0)-\E{Y_i}{\tau(S_0^{(0)}(Y_i),0)}$.

The decoding error probability  $p_{MAP}$ of the MAP\ decoder can be estimated by running simulations using the proposed sequential decoding algorithm with very large  $L$. This enables one to derive the termination threshold $T$, which depends only on channel and code properties. Numeric results suggest that in the case of AWGN\ channel such  threshold function can be well approximated by 
\begin{align}
\label{mThresholdApprox}
T\approx -\min\left\{\frac{a_{\mathcal C}\sigma^2+b_{\mathcal C}}{\sigma^2},\frac{t_C}{\sigma^2}\right\}
\end{align}  
for some parameters $a_{\mathcal C}, b_{\mathcal C}, t_{\mathcal C}$, which depend on the code $\mathcal C$, and can be obtained by curve fitting techniques.

Let  $p_{seq}$ and $p_{T}$ be the codeword error probabilities
of the sequential decoding algorithm presented in Section \ref{sSeqDecAlg}, and the  algorithm, which additionally discards all paths $v_0^{n-1}$ with $M(v_0^{n-1},y_0^{n-1})<T$, respectively. Obviously, $p_{MAP}\leq p_{seq}\leq p_{T}$.

Then one obtains\begin{align*}
& p_{T} = P\set{\mu< T} + P\set{\mathcal E|\mu \geq T} \leq 
   p_{MAP}+p_{seq}\leq 2p_{seq},
\end{align*}
where $\mathcal E$ is the event of sequential decoder error. 

Note that the proposed early termination method may abort the decoding even if $M_{3}( u_0^{n-1},y_0^{n-1})\geq T$, but at some $\phi$ it happens that $M_{3}(u_0^{\phi-1},y_0^{n-1})<T$, i.e. the probability of sequential decoding error with early termination is not less than $p_T$. However, the below presented numeric results suggest that the associated performance loss is negligible.

\section{Numeric results}
\label{sNumRes}
\begin{figure*}[th]
\centering
\includegraphics[width=0.7\textwidth]{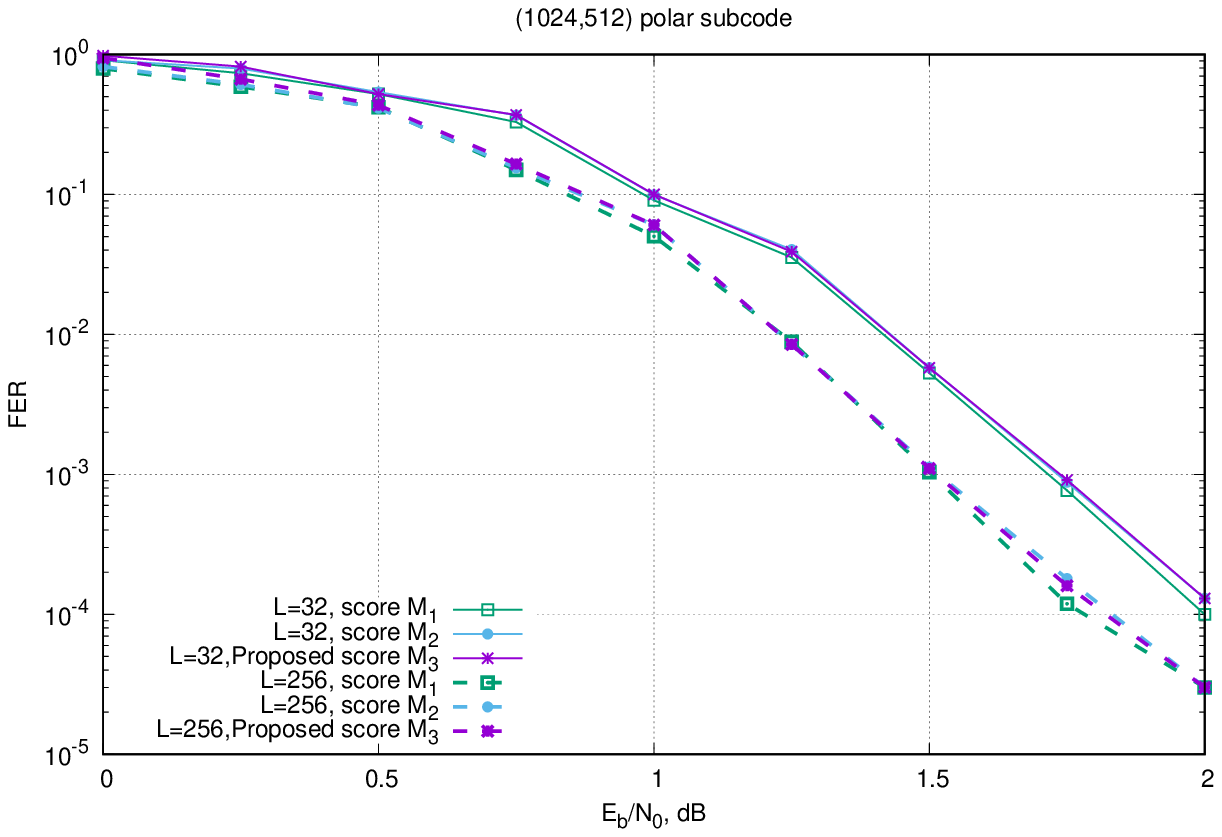}
\includegraphics[width=0.7\textwidth]{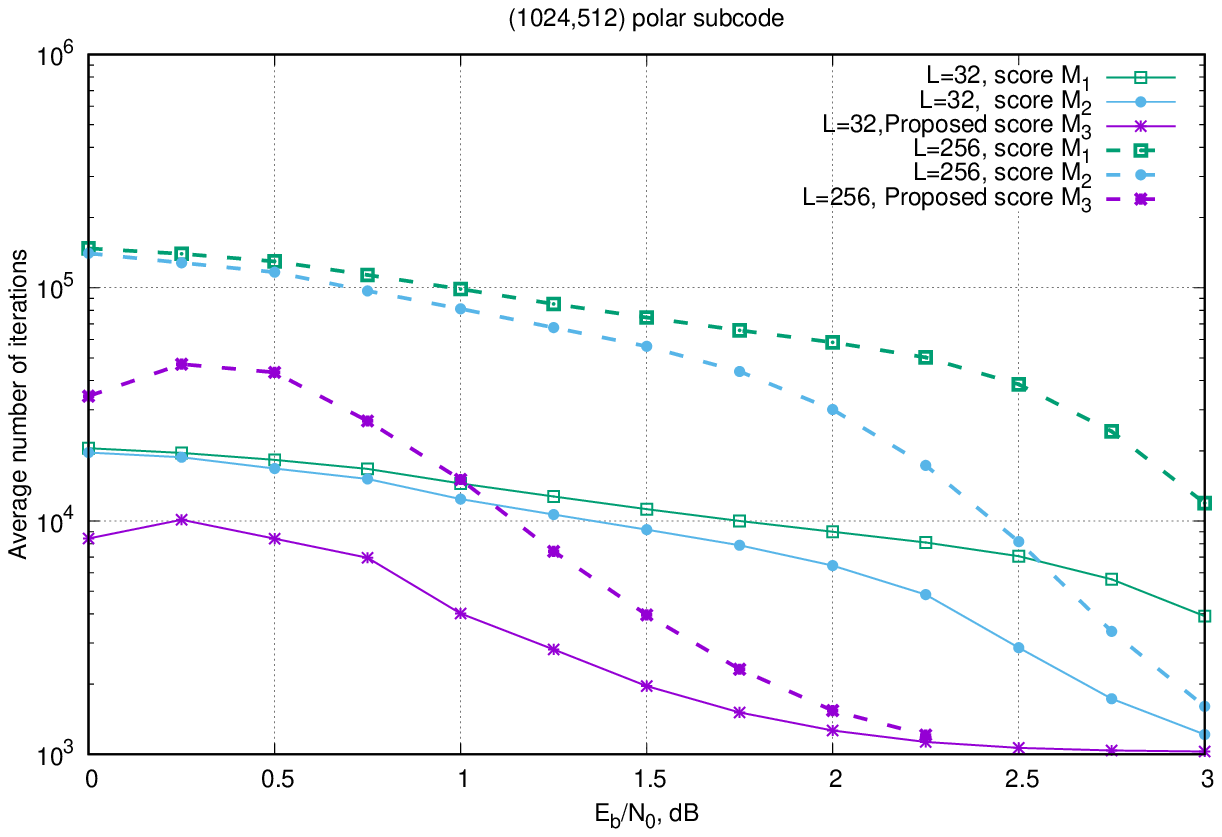}
\caption{The impact of the score function on the decoder performance and complexity. }
\label{fScorePerfComp}
\end{figure*}
\begin{table}
\caption{Average decoding complexity of $(1024,512,28)$ code with $L=32$, $\times 10^3$ operations}
\label{t1024512Compl}\centering
\begin{tabular}{|c|p{0.07\textwidth}|p{0.07\textwidth}|p{0.07\textwidth}|p{0.07\textwidth}|}\hline
\multirow{2}{*}{$E_b/N_0$, dB}&\multicolumn{2}{c|}{Summations}&\multicolumn{2}{c|}{Comparisons}\\\cline{2-5}
&Proposed path score&Path score from \cite{miloslavskaya2014sequential}&Proposed path score&Path score from \cite{miloslavskaya2014sequential}\\\hline
0.5&63.2&133&122.5&218\\\hline
1&34.8&73&55.6&122\\\hline
1.5&16&32&21.9&54\\\hline
2&8.8&18&12.0&31\\\hline
\end{tabular}
\end{table}
\begin{figure*}[th]
\centering
\includegraphics[width=0.8\textwidth]{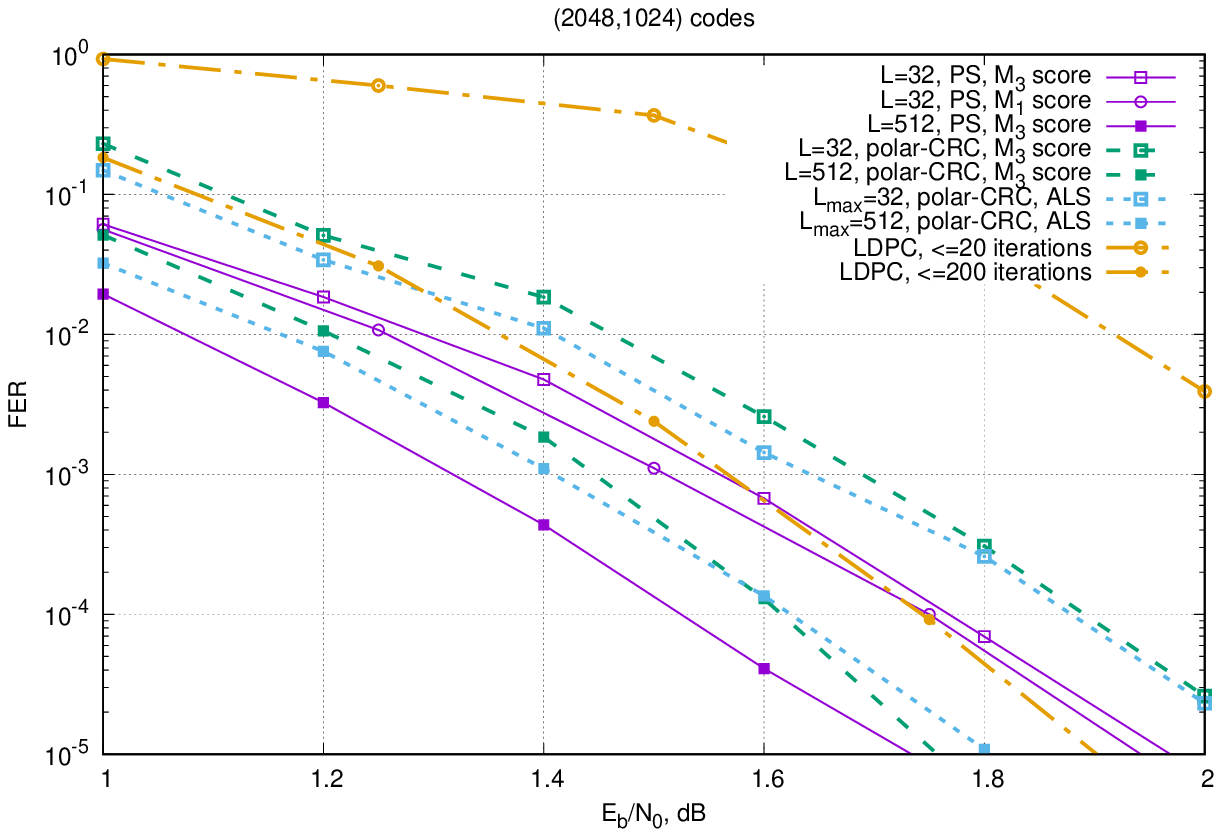} 
\includegraphics[width=0.8\textwidth]{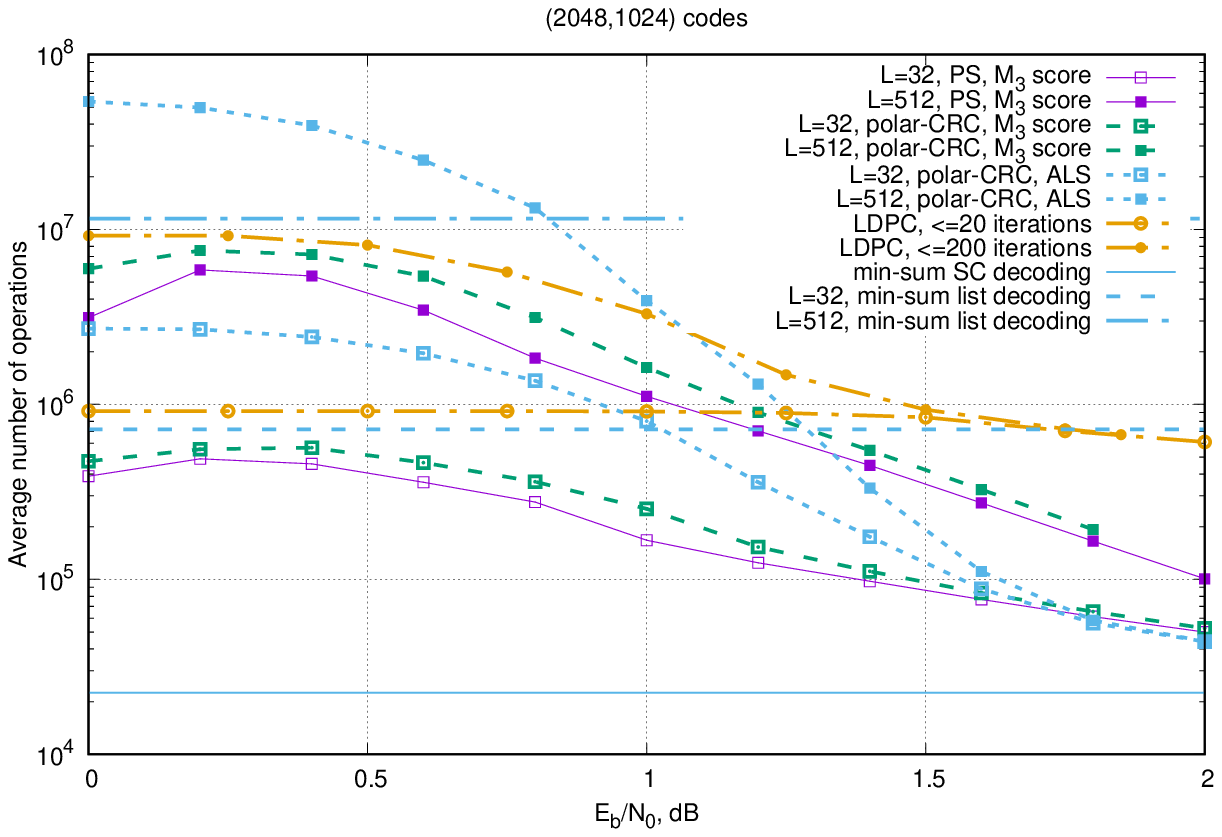}
\caption{Performance  and decoding   complexity  of $(2048,1024)$  codes.   }
\label{f2048}
\end{figure*}

\begin{figure*}[th]
\centering
\includegraphics[width=0.7\textwidth]{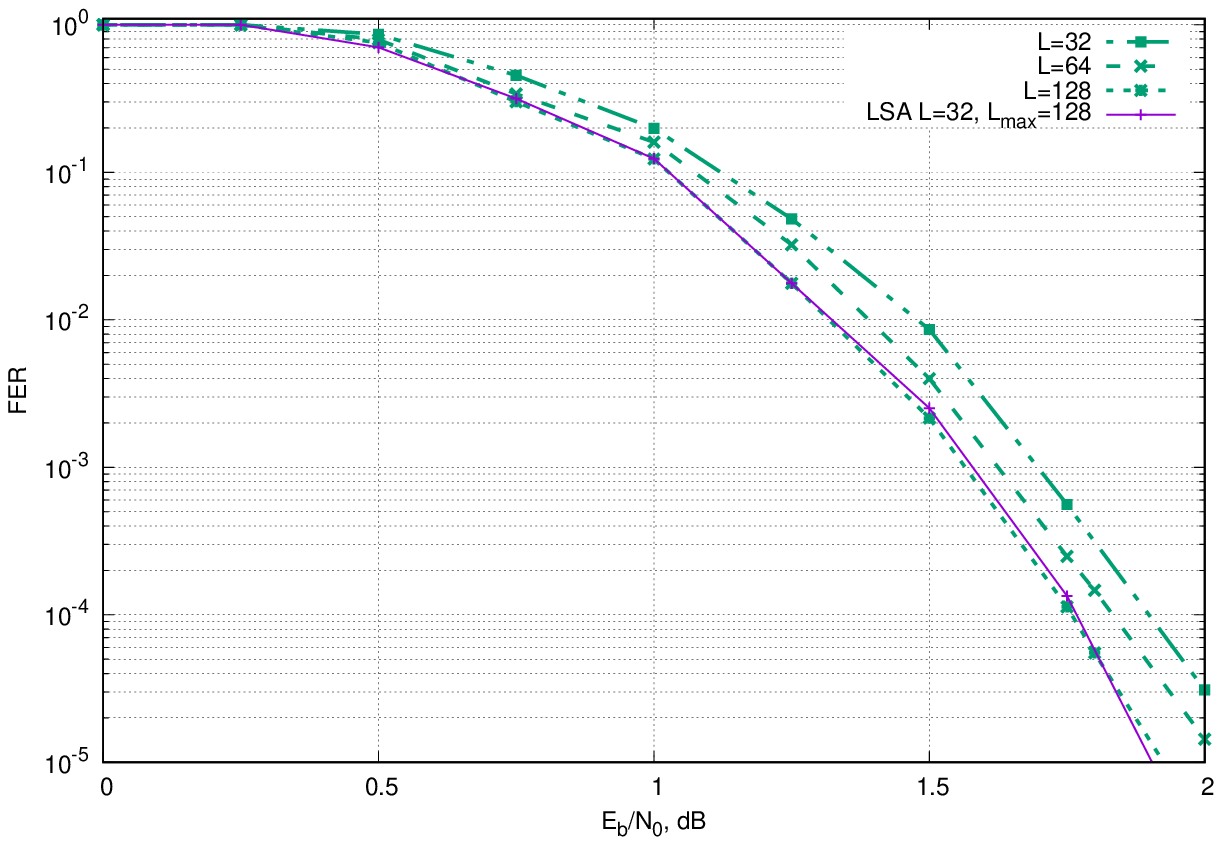}
\includegraphics[width=0.7\textwidth]{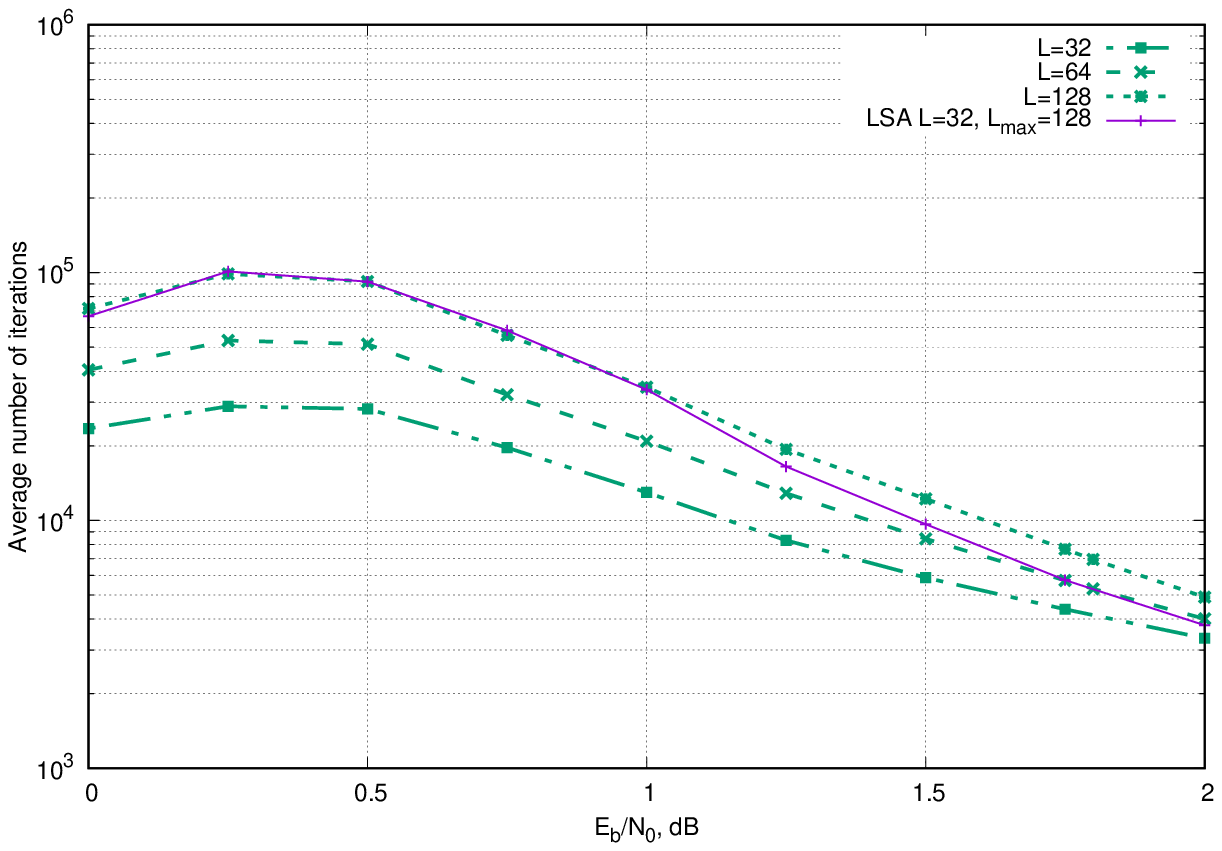}
\caption{Performance and  complexity    of sequential decoding with list  size   adaptation for $(2048,1024)$ polar subcode. }
\label{fALS}
\end{figure*}

The performance and complexity of the proposed decoding algorithm were investigated in the case of  BPSK modulation and AWGN\ channel. The results are reported for polar codes with 16-bit CRC (polar-CRC) and randomized polar subcodes\footnote{In order to ensure reproducibility of the results, we have set up a web site \url{http://dcn.icc.spbstu.ru/index.php?id=polar&L=2} containing the specifications of the considered  polar subcodes.} (PS)  \cite{trifonov2017randomized} . The size of the priority queue was set in all cases to $D=Ln$. The complexity of the decoding algorithm is reported in terms of the average number of iterations and average number of arithmetic operations. The first measure enables one to assess the efficiency of path selection by the proposed score function, and the second reflects the practical complexity of the algorithm. 

Figure \ref{fScorePerfComp} illustrates the decoding error probability and average number of iterations performed by the sequential decoder (without list size adaptation  and early termination) for the case of the  path scores $M_1$, $M_2$ and  $M_{3}$. The first two scores correspond to the Niu-Chen stack decoding algorithm and its min-sum version \cite{niu2012stack}. 
Observe, that the Niu-Chen algorithm was shown to achieve exactly the same performance as the Tal-Vardy list decoding algorithm with the same value of $L$, provided that the size of the priority queue $D$ is sufficiently high.

It can be seen that employing score $M_2$ results in a marginal performance loss, but significant reduction of the average number of iterations performed by the decoder. This is due to existence of multiple paths $v_0^{n-1}$ with low probability $W_m^{(n-1)}\set{v_0^{n-1}|y_0^{n-1}}$, which add up (see \eqref{mTotalProb}) to non-negligible  probabilities $W_m^{(\phi-1)}\set{v_0^{\phi-1}|y_0^{n-1}}$. Hence, employing path score $M_1$ causes the decoder to inspect many incorrect paths $v_0^{\phi-1}$.
At sufficiently high SNR the most probable continuation $\mathbf V(v_0^{\phi-1})$ of a path extracted at some phase from the PQ with high probability satisfies  all freezing constraints, so that the value given by  $M_2$ score function turns out to be close to the final path score. This enables the decoder to avoid visiting many incorrect paths in the code tree. 

Even more significant complexity reduction is obtained if one employs the proposed path score $M_3$. The proposed path score enables one to correctly compare the probabilities of paths $v_0^{\phi-1}$ of different length $\phi$. This results in an order of magnitude reduction of the average number of iterations.  Observe that the performance of the decoder employing the proposed score $M_{3}$ is essentially the same as in the case of  score $M_2$.
 Table \ref{t1024512Compl} provides comparison of the average number of arithmetic operations performed by the sequential decoder (without early termination and list size adaptation) implementing the proposed path score function, and the one presented in our prior work \cite{miloslavskaya2014sequential}. It can be seen that employing the proposed score function  results substantially lower average decoding complexity.

Figure \ref{f2048} presents the performance and average decoding complexity of $(2048,1024)$ codes. For comparison, we report also the performance  of polar codes with CRC-16 under list decoding with adaptive list size (ALS) \cite{li2012adaptive},    a CCSDS LDPC code under belief propagation decoding, and the complexity of the min-sum implementation of the Tal-Vardy algorithm with fixed list size. The complexity is presented in terms of the average number of summations and comparison operations for polar (sub)codes, and average number of summations and evaluations of $\log\tanh(x/2)$ for the LDPC code.

 It can be seen that for the case of a polar code with CRC the performance loss of the sequential decoding algorithm with respect to the Niu-Chen algorithm and list decoding with adaptive list size is more significant than in the case of $(1024,512)$ code. Observe that in this case the decoder needs to perform iterations until vector $v_0^{n-1}$ with valid CRC is extracted from the PQ. Hence, in this case the assumption that the value of \eqref{mCodeProb} is dominated by the first term\footnote{Recall, that this assumption is only a probabilistic one.} may be invalid with high probability

 However, the performance loss is much less significant in the case of  polar subcodes. Observe that at high SNR the ALS decoding algorithm has slightly lower average complexity than the proposed one. However, it is not obvious how to use the ALS decoder in conjunction with polar subcodes, which provide substantially better performance than polar codes with CRC, since this algorithm relies on checking CRC of the obtained data vector in order to detect if another decoding attempt with larger list size is needed.
 In the low-SNR region, where the frame error rate is at least $10^{-3}$, the proposed algorithm has lower average complexity compared to the ALS one, and in the case of polar subcodes provides up to 0.1 dB performance gain. 
 
Observe also, that the average number of summation and comparison operations in the case of the proposed decoding algorithm quickly converges to the complexity of the min-sum SC algorithm.
It can be also seen  that  polar subcodes under the proposed sequential decoding algorithm with $L=32$ provide the performance comparable to the state-of-the-art LDPC code, and with larger $L$ far outperform it. The average complexity of the proposed algorithm turns out to be substantially less compared to that of  BP\ decoding.
 Observe also that reducing the maximal number of iterations for the BP algorithm results in 0.5 dB performance loss with almost no gain in complexity, while the sequential decoding algorithm enables much better performance-complexity tradeoffs. 
\begin{figure*}[th]
\centering
\includegraphics[width=0.7\textwidth]{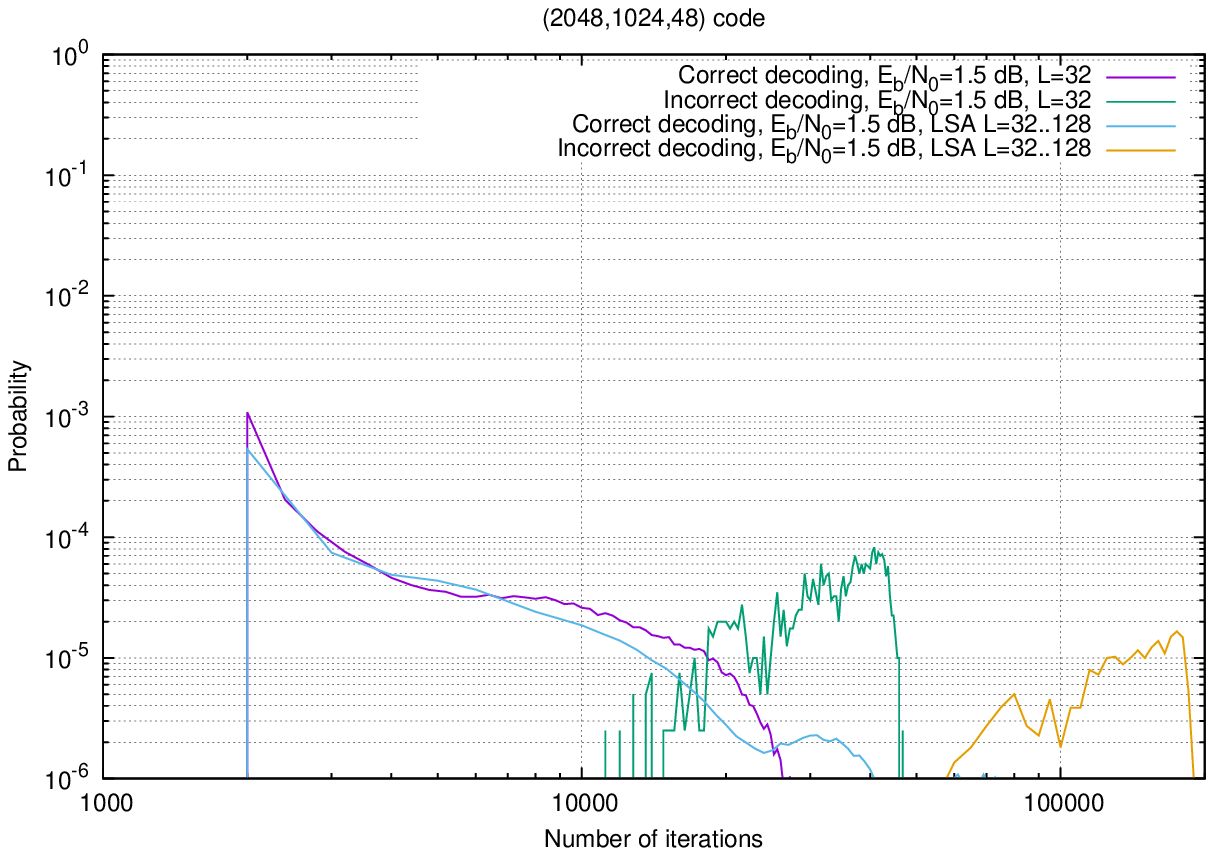}
\caption{Probability distribution of the number  of decoder iterations.}
\label{fDistr}
\end{figure*}

Figure \ref{fALS} illustrates the performance and complexity of the sequential decoding algorithm with score $M$ with list size adaptation (LSA) method described in Section \ref{sLSA}. Here list size $L$ was allowed to grow from $32$ to $L_{max}=128$. It was doubled after $\kappa_0=20$ iterations, such that at least one path was removed at step 6 of the algorithm described  in Section \ref{sSeqDecAlg}. 
 It can be seen that the proposed list size adaptation method enables one to achieve essentially the same performance as in the case of non-adaptive algorithm with $L=L_{max}$, but with lower average complexity. Observe that the proposed implementation of list size adaptation does not require restarting the decoder from scratch, as in the case of the  techniques considered in \cite{li2012adaptive,sarkis2016fast}.

Figure \ref{fDistr} illustrates the probability distribution of the number of iterations performed by the decoder in the case of correct and incorrect decoding. It can be seen that the distribution of the number of iterations in the event of correct decoding has rather heavy tail, and employing list size adaptation results in higher tail probabilities. In the event of decoding error the number of iterations may become quite high. However, there is very small probability that the decoding would be correct if the decoder has performed more than 20000 iterations. This can be also used for early termination of decoding. As it may be expected, employing list size adaptation results in heavier tail of the distribution for the case of incorrect decoding.

\begin{figure}[t]
\centering\includegraphics[width=0.5\textwidth]{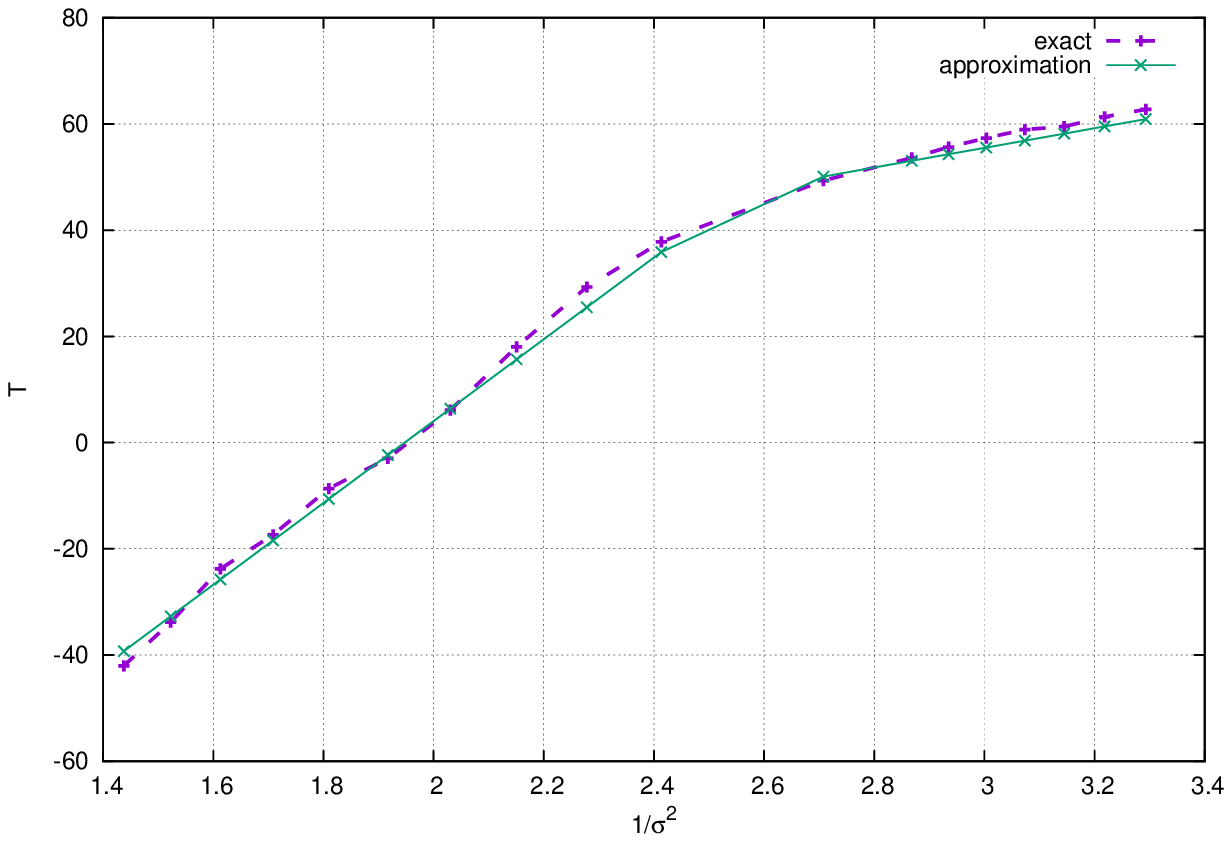}
\caption{Approximation     of the    early termination threshold for a $(1024, 736, 16)$ polar subcode. }
\label{fig:1024_736ts}
\end{figure}
\begin{table}[t]
\caption{Early termination parameters}
\label{t:et}
\centering
\begin{tabular}{|c| c |c |c|}
\hline
Code & $a_C$ & $b_C$ & $t_C$ \\
\hline
(1024,736) & -108.27 & 50.84 & 12 \\
\hline
(1024, 512) &  -116.37 & 121.41 & 43\\
\hline
\end{tabular} 
\end{table}

\begin{figure*}[th]
\centering
\includegraphics[width=0.7\textwidth]{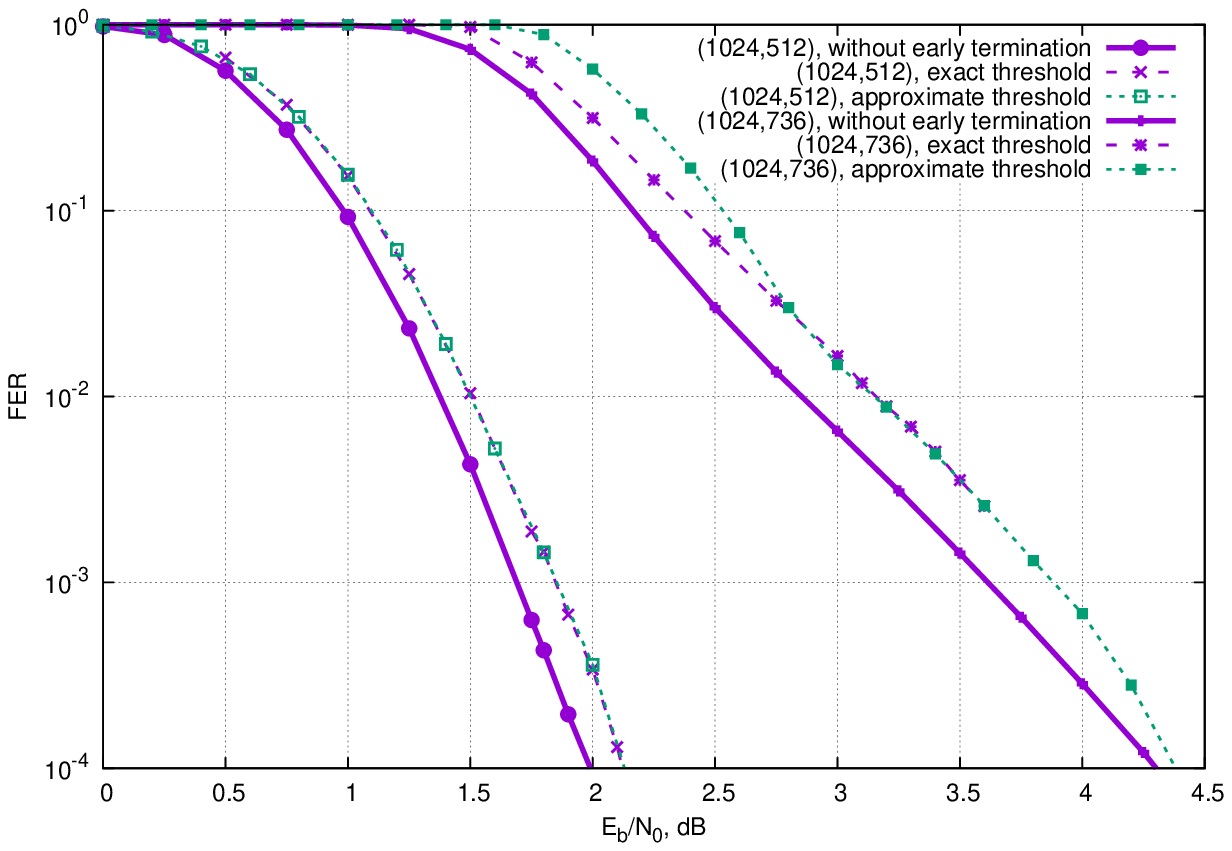}
\includegraphics[width=0.7\textwidth]{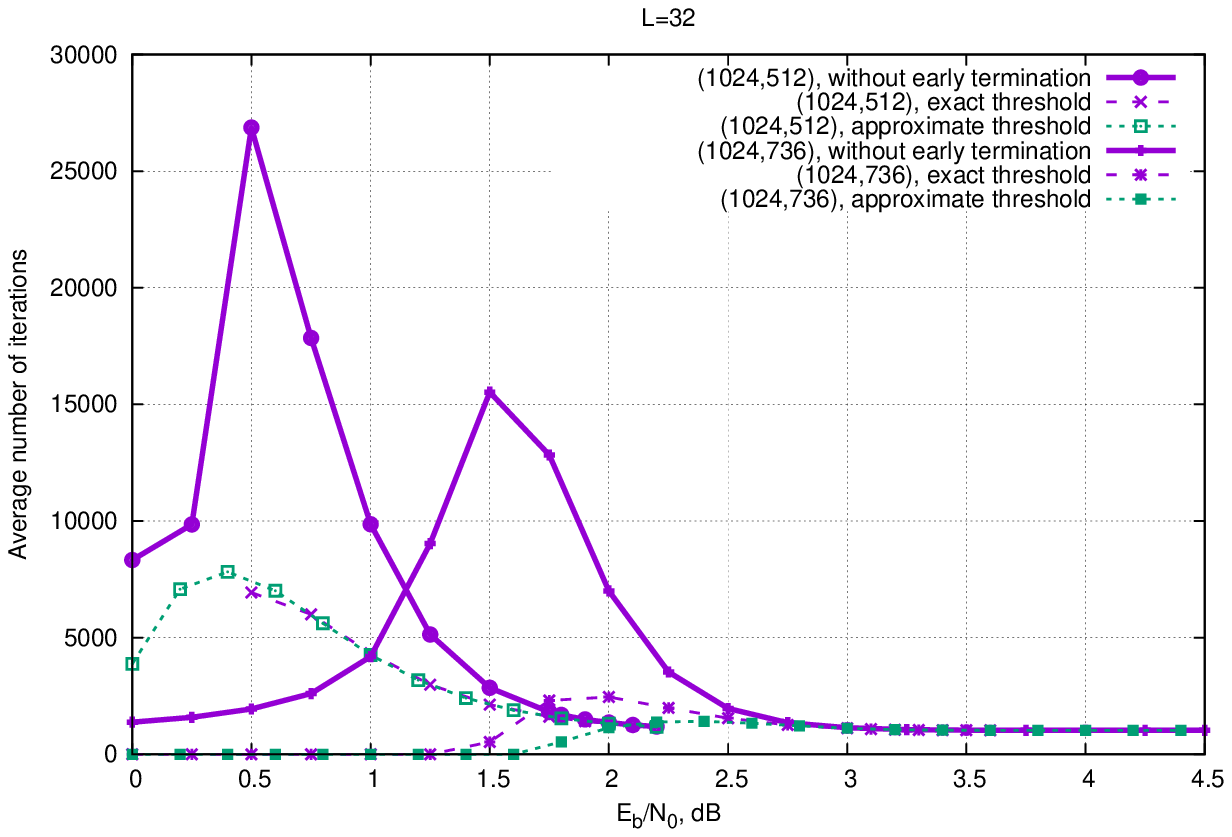}
\caption{Performance and complexity of sequential  decoding with  early termination. }
\label{fig:et}
\end{figure*}

Figure \ref{fig:1024_736ts} illustrates the termination threshold for the case of $(1024,736,16)$ polar subcode for different values of  $1/\sigma^2$. It can be seen that \eqref{mThresholdApprox} indeed represents a good approximation for $T$.

Figure \ref{fig:et} illustrates the performance and complexity of the sequential decoding algorithm with and without the proposed early termination method. The parameters of the early termination threshold function \eqref{mThresholdApprox} for the considered polar subcodes are given in Table \ref{t:et}.
It can be seen that the early termination condition enables one to significantly reduce the decoding complexity in the low-SNR\ region, where decoding error probability is high. This can be used to implement HARQ\ and adaptive coding protocols. Observe also, that the performance and complexity of the decoding algorithms employing the exact and approximate termination threshold functions are very close.

\section{Conclusions}
In this paper a novel decoding algorithm for polar (sub)codes was proposed. The proposed approach relies on the ideas of sequential decoding. The key contribution is a new path score function, which reflects the probability of the most likely continuation of a paths in the code tree, as well as the probability of satisfying already processed dynamic freezing constraints. The latter probability is difficult to compute exactly, so an approximation is proposed, which corresponds to the average behaviour of the correct path.  

 The proposed score function enables one to significantly reduce the average number of iterations performed by the decoder at the cost of a  negligible performance loss, compared to the case of the SCL decoding algorithm.  The worst-case complexity of the proposed decoding algorithm is $O(L n\log n)$, similarly to the case of the SCL\ algorithm. Furthermore, the proposed algorithm is based on the min-sum SC decoder, i.e. it can be implemented using only summation and comparison operations.

It was also shown that the performance of the proposed decoding algorithm can be substantially improved by recovering previously removed paths, and resuming their processing with increased list size. The improvement comes
at the cost of small increase of the average decoding complexity.
The average decoding complexity in the low-SNR region can be reduced by employing the proposed early termination method.


\begin{thebibliography}{10}
\providecommand{\url}[1]{#1}
\csname url@samestyle\endcsname
\providecommand{\newblock}{\relax}
\providecommand{\bibinfo}[2]{#2}
\providecommand{\BIBentrySTDinterwordspacing}{\spaceskip=0pt\relax}
\providecommand{\BIBentryALTinterwordstretchfactor}{4}
\providecommand{\BIBentryALTinterwordspacing}{\spaceskip=\fontdimen2\font plus
\BIBentryALTinterwordstretchfactor\fontdimen3\font minus
  \fontdimen4\font\relax}
\providecommand{\BIBforeignlanguage}[2]{{%
\expandafter\ifx\csname l@#1\endcsname\relax
\typeout{** WARNING: IEEEtran.bst: No hyphenation pattern has been}%
\typeout{** loaded for the language `#1'. Using the pattern for}%
\typeout{** the default language instead.}%
\else
\language=\csname l@#1\endcsname
\fi
#2}}
\providecommand{\BIBdecl}{\relax}
\BIBdecl

\bibitem{arikan2009channel}
E.~Arikan, ``Channel polarization: A method for constructing capacity-achieving
  codes for symmetric binary-input memoryless channels,'' \emph{IEEE
  Transactions on Information Theory}, vol.~55, no.~7, pp. 3051--3073, July
  2009.

\bibitem{tal2015list}
I.~Tal and A.~Vardy, ``List decoding of polar codes,'' \emph{IEEE Transactions
  On Information Theory}, vol.~61, no.~5, pp. 2213--2226, May 2015.

\bibitem{trifonov2016polar}
P.~Trifonov and V.~Miloslavskaya, ``Polar subcodes,'' \emph{IEEE Journal on
  Selected Areas in Communications}, vol.~34, no.~2, pp. 254--266, February
  2016.

\bibitem{trifonov2017randomized}
P.~Trifonov and G.~Trofimiuk, ``A randomized construction of polar subcodes,''
  in \emph{Proceedings of IEEE International Symposium on Information Theory},
  2017, pp. 1863--1867.

\bibitem{niu2012stack}
K.~Niu and K.~Chen, ``Stack decoding of polar codes,'' \emph{Electronics
  Letters}, vol.~48, no.~12, pp. 695--697, June 2012.

\bibitem{trifonov2013polar}
P.~Trifonov and V.~Miloslavskaya, ``Polar codes with dynamic frozen symbols and
  their decoding by directed search,'' in \emph{Proceedings of IEEE Information
  Theory Workshop}, September 2013, pp. 1 -- 5.

\bibitem{miloslavskaya2014sequential}
V.~Miloslavskaya and P.~Trifonov, ``Sequential decoding of polar codes,''
  \emph{IEEE Communications Letters}, vol.~18, no.~7, pp. 1127--1130, 2014.

\bibitem{Zigangirov1966some}
K.~S. Zigangirov, ``Some sequential decoding procedures,'' \emph{Problems of
  Information Transmission}, vol.~2, no.~4, pp. 1--10, 1966, in Russian.

\bibitem{johannesson1998fundamentals}
R.~Johannesson and K.~Zigangirov, \emph{Fundamentals of Convolutional
  Coding}.\hskip 1em plus 0.5em minus 0.4em\relax IEEE Press, 1998.

\bibitem{Cormen2001introduction}
T.~H. Cormen, C.~E. Leiserson, R.~L. Rivest, and C.~Stein, \emph{Introduction
  to Algorithms}, 2nd~ed.\hskip 1em plus 0.5em minus 0.4em\relax The MIT Press,
  2001.

\bibitem{massey1972variable}
J.~Massey, ``Variable-length codes and the {Fano} metric,'' \emph{IEEE
  Transactions on Information Theory}, vol.~18, no.~1, pp. 196--198, January
  1972.

\bibitem{balatsoukasstimming2015llrbased}
A.~Balatsoukas-Stimming, M.~B. Parizi, and A.~Burg, ``{LLR}-based successive
  cancellation list decoding of polar codes,'' \emph{IEEE Transactions On
  Signal Processing}, vol.~63, no.~19, pp. 5165--5179, October 2015.

\bibitem{leroux2013semiparallel}
C.~Leroux, A.~J. Raymond, G.~Sarkis, and W.~Gross, ``A semi-parallel
  successive-cancellation decoder for polar codes,'' \emph{IEEE Transactions on
  Signal Processing}, vol.~61, no.~2, pp. 289--299, January 2013.

\bibitem{leroux2011hardware}
C.~Leroux, I.~Tal, A.~Vardy, and W.~Gross, ``Hardware architectures for
  successive cancellation decoding ofpolar codes,'' in \emph{Proceedings of
  IEEE International Conference on Acoustics, Speech and Signal Processing},
  May 2011, pp. 1665--1668.

\bibitem{Valembois2004box}
A.~Valembois and M.~Fossorier, ``Box and match techniques applied to
  soft-decision decoding,'' \emph{IEEE Transactions on Information Theory},
  vol.~50, no.~5, pp. 796--810, May 2004.

\bibitem{moorthy1997softdecision}
H.~T. Moorthy, S.~Lin, and T.~Kasami, ``Soft-decision decoding of binary linear
  block codes based on an iterative search algorithm,'' \emph{IEEE Transactions
  On Information Theory}, vol.~43, no.~3, pp. 1030--1040, May 1997.

\bibitem{kern2014new}
D.~Kern, S.~Vorkoper, and V.~Kuhn, ``A new code construction for polar codes
  using min-sum density,'' in \emph{Proceedings of International Symposium on
  Turbo Codes and Iterative Information Processing}, 2014, pp. 228--232.

\bibitem{li2012adaptive}
B.~Li, H.~Shen, and D.~Tse, ``An adaptive successive cancellation list decoder
  for polar codes with cyclic redundancy check,'' \emph{IEEE Communications
  Letters}, vol.~16, no.~12, pp. 2044--2047, December 2012.

\bibitem{sarkis2016fast}
G.~Sarkis, P.~Giard, A.~Vardy, C.~Thibeault, and W.~Gross, ``Fast list decoders
  for polar codes,'' \emph{IEEE Journal On Selected Areas In Communications},
  vol.~34, no.~2, pp. 318--328, February 2016.

\end{thebibliography}

\end{document}